\documentclass[11pt]{article}
\usepackage{amsmath,amssymb,amsthm,amsfonts,fullpage,stackrel}
\usepackage{MnSymbol}

\newcommand{\map}[3]{{{#1}:{#2}\rightarrow{#3}}}

\newcommand{\complexes}{\mathbb{C}}
\newcommand{\cH}{\mathcal{H}}
\newcommand{\bra}[1]{{\langle{#1}|}}

\newcommand{\ket}[1]{{|{#1}\rangle}}

\newcommand{\braket}[2]{{\langle{#1}|{#2}\rangle}}
\newcommand{\matelt}[3]{{\langle{#1}|{#2}|{#3}\rangle}}
\newcommand{\disjunion}{\cupdot}
\newcommand{\p}{\varphi}
\newcommand{\1}{\mathbf{1}}
\newcommand{\csign}{\mathit{CZ}}

\theoremstyle{definition}
\newtheorem{definition}{Definition}[section]

\newtheorem{remark}[definition]{Remark}

\theoremstyle{plain}
\newtheorem{theorem}[definition]{Theorem}
\newtheorem{lemma}[definition]{Lemma}

\title{A note on the entangling properties of the C-SIGN and related quantum gates}
\author{Stephen A. Fenner\thanks{Computer Science and Engineering Department, Columbia, SC 29208 USA. \texttt{fenner.sa@gmail.com}}\\University of South Carolina}

\begin{document}

\maketitle

\begin{abstract}
We show that any essential application of an $n$-qubit C-SIGN or related quantum gate $G$ leaves its qubits everywhere entangled, provided they were not everywhere entangled to begin with.  By ``essential'' we mean roughly that $G$ is not equivalent to a similar gate of smaller arity (or to the identity) when applied to the same input state.  By ``everywhere entangled'' we mean that the state is not separable with respect to any bipartition of the qubits.
\end{abstract}

\section{Preliminaries}

Following standard practice, we let $[n]$ denote the set $\{1,\ldots,n\}$ for any integer $n \ge 0$.  We write $z^*$ for the complex conjugate of a complex number $z$, and we write $A^*$ for the adjoint (Hermitian conjugate) of an operator $A$ on a Hilbert space.  Otherwise, our notation is fairly standard (see \cite{KLM:quantum-book,KSV:quantum-book,NC:quantumbook} for example).

The C-SIGN gate on $k$ qubits is a unitary operator $\csign$ defined thus for every computational basis state $\ket{x_1x_2\cdots x_k}$, for $x_1,x_2,\ldots,x_k\in\{0,1\}$:
\[ \csign\ket{x_1x_2\cdots x_k} = (-1)^{x_1\cdots x_k}\ket{x_1x_2\cdots x_k}\;. \]
Generalizing this definition a bit, for any $\eta\in\complexes$ such that $|\eta| = 1$ and $\eta \ne 1$, we define the unitary gate $G_\eta$ by
\[ G_\eta\ket{x_1x_2\cdots x_k} = \eta^{x_1\cdots x_k}\ket{x_1x_2\cdots x_k}\;. \]
$G_\eta$ is represented in the computational basis by a diagonal matrix, and it has two eigenspaces: the one-dimensional subspace $E :=\{a\ket{1\cdots 1} : a\in\complexes\}$ with eigenvalue $\eta$, and its orthogonal complement $E^\perp$ with eigenvalue $1$.  $E^\perp$ is spanned by those basis vectors with at least one $0$ in their corresponding strings.  Note that $G_\eta$ is unitary, that $G_\eta^* = G_{\eta^*}$, and that $G_\eta$ commutes with the swap operator on any pair of its qubits.

We now fix $\cH$ to be the $n$-qubit Hilbert space, for some $n > 0$.  We let the qubits of $\cH$ have indices from $1$ to $n$.  The computational basis of $\cH$ is thus $\left\{ \ket{x} \mid \map{x}{[n]}{\{0,1\}}\right\}$, indexed by binary strings of length $n$.  Fix any subset $S\subseteq [n]$.  We let $\cH_S$ denote the Hilbert space of the qubits in $S$ (or, more strictly speaking, the qubits whose indices are in $S$).  So for example, $\cH \cong \cH_S \otimes \cH_{\overline S}$, where we write $\overline S$ for $[n]\setminus S$.  Similarly, if $\map{x}{[n]}{\{0,1\}}$ is any length-$n$ binary string, we let $x_{|S}$ denote the restriction of $x$ to $S$, and for $i\in [n]$ we write $x_i$ for $x_{|\{i\}}$.  We use the term, ``string'' to refer generally to $0,1$-valued maps whose domains are arbitrary subsets of $[n]$.  If we do not specify the domain of a string, we assume it is $[n]$.

We let $\1$ denote the string of $n$ many $1$'s, i.e., the constant $1$-valued string with domain $[n]$.

If $\map{y}{J}{\{0,1\}}$ and $\map{z}{K}{\{0,1\}}$ are strings for disjoint sets $J,K\subseteq [n]$, then we write $y\cup z$ for the unique string with domain $J\cup K$ extending $y$ and $z$.  So in particular, for any $S\subseteq [n]$, if $\ket{y}$ is a computational basis state of $\cH_S$ and $\ket{z}$ is a computational basis state of $\cH_{\overline S}$, then $\ket{y\cup z}$ is the computational basis state of $\cH$ corresponding to $\ket{y}\otimes\ket{z}$.

\begin{definition}
Let $S\subseteq [n]$ and let $\ket{\psi}\in\cH$ be any unit vector.  We say that $\ket{\psi}$ is \emph{$S$-separable} if there exists a bipartition $[n] = A\disjunion B$ into sets $A$ and $B$ such that $A\cap S \ne \emptyset$ and $B\cap S \ne \emptyset$ and $\ket{\psi} = \ket{\psi}_A \otimes \ket{\psi}_B$ for some unit vectors $\ket{\psi}_A \in \cH_A$ and $\ket{\psi}_B \in \cH_B$.  Otherwise, we say that $\ket{\psi}$ is \emph{$S$-entangled}.
\end{definition}

\section{Main Result}

Our main result is Theorem~\ref{thm:main}, below.  A weaker version of this result was proved independently by Daniel Grier, Thomas Thierauf, and myself (via private communications).

We now fix for the entire sequel some arbitrary $\eta\in\complexes$ such that $|\eta|=1$ and $\eta\ne 1$.

\begin{definition}
For any set $S\subseteq [n]$, let $G_S$ be the $G_\eta$ gate applied to the qubits in $S$ (which means that $G_S$ is an operator on $\cH_S$).  Similarly, let $I_S$ be the identity operator applied to the qubits in $S$.  (If $S=\emptyset$, then $\cH_S$ has dimension~$1$ and we define $G_S := \eta I_S$ by convention.)
\end{definition}

\begin{definition}
Let $S\subseteq [n]$ be nonempty, and let $G := G_S\otimes I_{\overline{S}}$ (so $G$ is an operator on $\cH$).  Let $\ket{\psi}\in\cH$ be some unit vector.  We say that $G$ \emph{simplifies} on $\ket{\psi}$ if either $G\ket{\psi} = \ket{\psi}$ or $G\ket{\psi} = (G_{S'}\otimes I_{\overline{S'}})\ket{\psi}$ for some proper subset $S'\subset S$.
\end{definition}

There are two ways that $G$ can simplify on $\ket{\psi}$: either $\braket{x}{\psi} = 0$ for every string $x$ such that $x_{|S} = \1_{|S}$ (whence $G\ket{\psi} = \ket{\psi}$), or there exists $i\in S$ such that, for all strings $x$ with $x_i = 0$, we have $\braket{x}{\psi} = 0$.  In the former case, $\ket{\psi}$ is an eigenvector of $G$ with eigenvalue $1$ and so $G\ket{\psi} = \ket{\psi}$; every computational basis vector appearing in the expansion of $\ket{\psi}$ (as a linear combination of computational basis vectors) has a $0$ somewhere in $S$.  This $0$ turns off $G$ completely.  In the latter case, $G\ket{\psi} = (G_{S\setminus\{i\}} \otimes I_{\overline S \cup \{i\}})\ket{\psi}$, that is, $G$ acts the same as a smaller $G_\eta$ gate applied to all qubits in $S$ except the $i^\text{th}$.

\begin{theorem}\label{thm:main}
Let $S\subseteq [n]$ have size at least~$2$, and let $G := G_S\otimes I_{\overline{S}}$.  Let $\ket{\psi}\in\cH$ be any unit vector.  Then at least one of the following is true: (1) $\ket{\psi}$ is $S$-entangled; (2) $G\ket{\psi}$ is $S$-entangled; or (3) $G$ simplifies on $\ket{\psi}$.
\end{theorem}

\begin{proof}
Let $\ket{\p} := G\ket{\psi}$.
Since $G$ is represented by a diagonal matrix, for any string $x$, we have $|\braket{x}{\p}| = |\matelt{x}{G}{\psi}| = |\braket{x}{\psi}|$, so in particular, $\braket{x}{\p} = 0$ if and only if $\braket{x}{\psi} = 0$.

Suppose $\ket{\psi}$ and $\ket{\p}$ are both $S$-separable.  Write $\ket{\psi} = \ket{\psi}_A \otimes \ket{\psi}_B$, where $A\disjunion B = [n]$, $A$ and $B$ each have nonempty intersection with $S$, and $\ket{\psi}_A\in\cH_A$ and $\ket{\psi}_B\in\cH_B$ are unit vectors.  Likewise, write $\ket{\p} = \ket{\p}_C \otimes \ket{\p}_D$, for $C$ and $D$ where $C\disjunion D = [n]$, each have nonempty intersection with $S$, and $\ket{\p}_C\in\cH_C$ and $\ket{\p}_D\in\cH_D$ are unit vectors.

Now assume for the sake of contradiction that $G$ does not simplify on $\ket{\psi}$.  Then we have $G\ket{\psi} \ne \ket{\psi}$, and so there exists a string $u$ such that $u_{|S} = \1_{|S}$ and $\braket{u}{\psi} \ne 0$.  Fix such a $u$, noting that $G\ket{u} = \eta\,\ket{u}$.


We say that a string $\map{x}{[n]}{\{0,1\}}$ is a \emph{test string} if, for every nonempty $Y \in \{S\cap A\cap C, S\cap A\cap D, S\cap B\cap C, S\cap B\cap D\}$, there exists $i\in Y$ such that $x_i = 0$.  We will derive a contradiction in two steps: (1) show that $\braket{x}{\psi} = 0$ for every test string $x$; and (2) construct a test string $y$ such that $\braket{y}{\psi} \ne 0$.

To show step~(1), fix an arbitrary test string $x$.  We first chop $x$ into two parts in two different ways: (1) $x_{|A}$ and $x_{|B}$; (2) $x_{|C}$ and $x_{|D}$.  Each pair unions to $x$.  From $x_{|A}$ we get four strings $\map{x^A_{jk}}{A}{\{0,1\}}$ for $j,k\in\{0,1\}$ by changing some $0$-entries in $x_{|A}$ to $1$: Define
\begin{align*}
x^A_{00} &:= x_{|A}\;, &
x^A_{01} &:= x_{|A\cap C} \cup u_{|A\cap D}\;, \\
x^A_{10} &:= u_{|A\cap C} \cup x_{|A\cap D}\;, &
x^A_{11} &:= u_{|A}\;.
\end{align*}
We make similar definitions using $x_{|B}$, $x_{|C}$, and $x_{|D}$ with domains $B$, $C$, and $D$, respectively: Define
\begin{align*}
x^B_{00} &:= x_{|B}\;, &
x^B_{01} &:= x_{|B\cap C} \cup u_{|B\cap D}\;, \\
x^B_{10} &:= u_{|B\cap C} \cup x_{|B\cap D}\;, &
x^B_{11} &:= u_{|B}\;, \\ \\
x^C_{00} &:= x_{|C}\;, &
x^C_{01} &:= x_{|C\cap A} \cup u_{|C\cap B}\;, \\
x^C_{10} &:= u_{|C\cap A} \cup x_{|C\cap B}\;, &
x^C_{11} &:= u_{|C}\;, \\ \\
x^D_{00} &:= x_{|D}\;, &
x^D_{01} &:= x_{|D\cap A} \cup \1_{|D\cap B}\;, \\
x^D_{10} &:= u_{|D\cap A} \cup x_{|D\cap B}\;, &
x^D_{11} &:= u_{|D}\;.
\end{align*}
There are two things to observe about these definitions:
\begin{enumerate}
\item
We have $x^A_{00}\cup x^B_{00} = x^C_{00}\cup x^D_{00} = x$.
\item
For all $j,k,\ell,m\in\{0,1\}$,
\begin{equation}\label{eqn:ABCD-strings}
x^A_{jk}\cup x^B_{\ell m} = x^C_{j\ell}\cup x^D_{km}\;.
\end{equation}
For example, for all $i\in [n]$, we have
\[ (x^A_{00}\cup x^B_{10})_i = (x^C_{01}\cup x^D_{00})_i = \left\{ \begin{array}{ll}
u_i & \mbox{if $i \in B\cap C$,} \\
x_i & \mbox{otherwise.}
\end{array} \right. \]
\end{enumerate}

We now consider only the coefficients in $\ket{\psi}_A$, $\ket{\psi}_B$, $\ket{\psi}_C$, and $\ket{\psi}_D$ of the basis vectors given above.  For all $j,k\in\{0,1\}$, define
\begin{align*}
a_{jk} &:= \braket{x^A_{jk}}{\psi}_A &\mbox{(scalar product in $\cH_A$),} \\
b_{jk} &:= \braket{x^B_{jk}}{\psi}_B &\mbox{(scalar product in $\cH_B$),} \\
c_{jk} &:= \braket{x^C_{jk}}{\p}_C &\mbox{(scalar product in $\cH_C$),} \\
d_{jk} &:= \braket{x^D_{jk}}{\p}_D &\mbox{(scalar product in $\cH_D$).}
\end{align*}
For example, $a_{jk}$ is the coefficient of $\ket{x^A_{jk}}$ in the expansion of $\ket{\psi}_A$ in terms of basis vectors in $\cH_A$.  (The $a_{jk}$, $b_{jk}$, $c_{jk}$, and $d_{jk}$ may depend on the particular choice of test string $x$.)

Recalling that $\ket{\psi} = \ket{\psi}_A \otimes \ket{\psi}_B$ and $\ket{\p} = \ket{\p}_C \otimes \ket{\p}_D$, we get, for all $j,k,\ell,m\in\{0,1\}$,
\begin{align}\label{eqn:case-1-AB}
\braket{x^A_{jk}\cup x^B_{\ell m}}{\psi} &= (\bra{x^A_{jk}}\otimes\bra{x^B_{\ell m}})(\ket{\psi}_A \otimes \ket{\psi}_B) = \braket{x^A_{jk}}{\psi}_A\braket{x^B_{\ell m}}{\psi}_B = a_{jk}b_{\ell m}\;, \\ \label{eqn:case-1-CD}
\braket{x^C_{jk}\cup x^D_{\ell m}}{\p} &= (\bra{x^C_{jk}}\otimes\bra{x^D_{\ell m}})(\ket{\p}_C \otimes \ket{\p}_D) = \braket{x^C_{jk}}{\p}_A\braket{x^D_{\ell m}}{\p}_B = c_{jk}d_{\ell m}\;.
\end{align}
Then by observation~(1), if we can show that $a_{00}b_{00} = 0$, then $\braket{x}{\psi} = a_{00}b_{00} = 0$ for any test string $x$.

For any string $\map{x}{[n]}{\{0,1\}}$, if there exists $i\in S$ such that $x_i = 0$, then $G\ket{x} = \ket{x}$, and if $x_i = 1$ for all $i\in S$, then $G\ket{x} = \eta\,\ket{x}$.  This fact gives us equations among the $a_{jk},b_{jk},c_{jk},d_{jk}$ by comparing amplitudes in $\ket{\psi}$ versus $\ket{\p}$.  Which equations we get depends on which of the sets $S\cap A\cap C$, $S\cap A\cap D$, $S\cap B\cap C$, and $S\cap B\cap D$ are empty.  At most two of these sets can be empty, so we have three cases.

\paragraph{Case~1.} $S\cap A\cap C$, $S\cap A\cap D$, $S\cap B\cap C$, and $S\cap B\cap D$ are all nonempty.

In this case, $x^C_{11}\cup x^D_{11} = u$, and if $jk\ell m = 0$ then $x^A_{j\ell}\cup x^D_{km}$ has a $0$ somewhere in $S$.  Using observation~(2) above, we then get
\begin{equation}\label{eqn:case-1-G-on-basis}
G\ket{x^A_{jk}\cup x^B_{\ell m}} = G\ket{x^C_{j\ell}\cup x^D_{km}} = \left\{ \begin{array}{ll}
\eta\,\ket{x^C_{j\ell}\cup x^D_{km}} & \mbox{if $j=k=\ell=m=1$,} \\
\ket{x^C_{j\ell}\cup x^D_{km}} & \mbox{if $jk\ell m = 0$.}
\end{array}\right.
\end{equation}
Then combining Equations~(\ref{eqn:ABCD-strings},\ref{eqn:case-1-AB},\ref{eqn:case-1-CD},\ref{eqn:case-1-G-on-basis}) and the fact that $\ket{\p} = G\ket{\psi}$, we get 16 equations: for all $j,k,\ell,m\in\{0,1\}$,
\begin{equation}\label{eqn:case-1}
c_{j\ell}d_{km} = \left\{ \begin{array}{ll}
\eta\,a_{jk}b_{\ell m} & \mbox{if $j=k=\ell=m=1$,} \\
a_{jk}b_{\ell m} & \mbox{if $jk\ell m = 0$.}
\end{array}\right.
\end{equation}

By assumption, $\braket{u}{\psi} \ne 0$, and so $\braket{u}{\psi} = a_{11}b_{11} \ne 0$, and $c_{11}d_{11} = \eta\,a_{11}b_{11} \ne 0$ as well.  This fact together with Equation~(\ref{eqn:case-1}) implies
$a_{00}b_{00} = 0$ by Lemma~\ref{lem:4-sets} in Appendix~\ref{sec:calculations}.

\paragraph{Case~2.}  One of $S\cap A\cap C$, $S\cap A\cap D$, $S\cap B\cap C$, and $S\cap B\cap D$ is empty and the other three are nonempty.  Without loss of generality, we assume that $S\cap B\cap C = \emptyset$.

In this case, $x^C_{j\ell}\cup x^D_{km} = u$ if $j = k = m = 1$ (independent of $\ell$, because the test string $x$ has no $0$ in $S\cap B\cap C$), and otherwise if $jkm = 0$, we get that $x^A_{j\ell}\cup x^D_{km}$ has a $0$ somewhere in $S$.  Thus
\begin{equation}\label{eqn:case-2-G-on-basis}
G\ket{x^A_{jk}\cup x^B_{\ell m}} = G\ket{x^C_{j\ell}\cup x^D_{km}} = \left\{ \begin{array}{ll}
\eta\,\ket{x^C_{j\ell}\cup x^D_{km}} & \mbox{if $j=k=m=1$,} \\
\ket{x^C_{j\ell}\cup x^D_{km}} & \mbox{if $jkm = 0$.}
\end{array}\right.
\end{equation}
Then setting $\ell := 0$ we get eight equations: for all $j,k,m\in\{0,1\}$,
\begin{equation}\label{eqn:case-2}
c_{j0}d_{km} = \left\{ \begin{array}{ll}
\eta\,a_{jk}b_{0m} & \mbox{if $j=k=m=1$,} \\
a_{jk}b_{0m} & \mbox{if $jkm = 0$.}
\end{array}\right.
\end{equation}
These equations again imply $a_{00}b_{00} = 0$ by Lemma~\ref{lem:3-sets} in Appendix~\ref{sec:calculations}.

\paragraph{Case~3.}  Two of $S\cap A\cap C$, $S\cap A\cap D$, $S\cap B\cap C$, and $S\cap B\cap D$ are empty.  Without loss of generality, we assume that $S\cap A\cap D = S\cap B\cap C = \emptyset$, whence $S\cap A = S\cap C$ and $S\cap B = S\cap D$, and both are nonempty.  We argue analogously to Cases~1 and 2.

In this case, $x^C_{j\ell}\cup x^D_{km} = u$ if $j = m = 1$ (independent of $k$ and $\ell$), and otherwise if $jm = 0$, we get that $x^A_{j\ell}\cup x^D_{km}$ has a $0$ somewhere in $S$.  Thus
\begin{equation}\label{eqn:case-3-G-on-basis}
G\ket{x^A_{jk}\cup x^B_{\ell m}} = G\ket{x^C_{j\ell}\cup x^D_{km}} = \left\{ \begin{array}{ll}
\eta\,\ket{x^C_{j\ell}\cup x^D_{km}} & \mbox{if $j=m=1$,} \\
\ket{x^C_{j\ell}\cup x^D_{km}} & \mbox{if $jm = 0$.}
\end{array}\right.
\end{equation}
Then setting $k := \ell := 0$ we get four equations: for all $j,m\in\{0,1\}$,
\begin{equation}\label{eqn:case-3}
c_{j0}d_{0m} = \left\{ \begin{array}{ll}
\eta\,a_{j0}b_{0m} & \mbox{if $j=m=1$,} \\
a_{j0}b_{0m} & \mbox{if $jm = 0$.}
\end{array}\right.
\end{equation}
These equations also imply $a_{00}b_{00} = 0$ by Lemma~\ref{lem:2-sets} in Appendix~\ref{sec:calculations}.

This establishes step~(1) in the contradiction proof.

For step~(2), we now construct a test string $y$ such that $\braket{y}{\psi} \ne 0$.  We first show the construction assuming Case~1 above, then modify it slightly for Cases~2 and 3.

Assume Case~1.  Choose some $i\in S\cap A\cap C$.  Since $G$ does not simplify on $\ket{\psi}$, there exists a string $y_{AC}$ (with domain $[n]$) such that $\braket{y_{AC}}{\psi} \ne 0$ and $(y_{AC})_i = 0$.  Then since $G$ fixes $\ket{y_{AC}}$, we have
\[ 0 \ne \braket{y_{AC}}{\psi} = \braket{y_{AC}}{\p} = \braket{(y_{AC})_{|C}\cup (y_{AC})_{|D}}{\p} = \braket{(y_{AC})_{|C}}{\p}_C\braket{(y_{AC})_{|D}}{\p}_D\;. \]
In particular, $\braket{(y_{AC})_{|C}}{\p}_C \ne 0$.  Now we can choose some string $y_{AD}$ such that $(y_{AD})_i = 0$ for some $i\in S\cap A\cap D$.  Analogously to the above, we get
\[ 0 \ne \braket{y_{AD}}{\psi} = \braket{y_{AD}}{\p} = \braket{(y_{AD})_{|C}
\cup(y_{AD})_{|D}}{\p} = \braket{(y_{AD})_{|C}}{\p}_C\braket{(y_{AD})_{|D}}{\p}_D\;. \]
In particular, $\braket{(y_{AD})_{|D}}{\p}_D \ne 0$.  Now define the string
\[ y_A := (y_{AC})_{|C} \cup (y_{AD})_{|D}\;. \]
Note that $(y_A)_i = (y_A)_j = 0$ for some $i\in S\cap A\cap C$ and $j\in S\cap A\cap D$.  Furthermore,
\[ \braket{y_A}{\psi} = \braket{y_A}{\p} = \braket{(y_{AC})_{|C}}{\p}_C \braket{(y_{AD})_{|D}}{\p}_D \ne 0\;. \]

By exactly repeating the argument in the previous paragraph with $B$ substituted for $A$, we obtain a string $y_B$ such that $(y_B)_i = (y_B)_j = 0$ for some $i\in S\cap B\cap C$ and $j\in S\cap B\cap D$, and furthermore, $\braket{y_B}{\psi} \ne 0$.

Finally, let $y := (y_A)_{|A} \cup (y_B)_{|B}$.  Observe that $y$ is a test string and that
\[ \braket{y}{\psi} = \braket{y_{|A}\cup y_{|B}}{\psi} = \braket{y_{|A}}{\psi}_A \braket{y_{|B}}{\psi}_B \ne 0\;. \]
This concludes the proof for Case~1.

Assume Case~2.  Using an identical construction to that of Case~1, we obtain a string $y_A$ such that $\braket{y_A}{\psi} \ne 0$ and $(y_A)_i = (y_A)_j = 0$ for some $i\in S\cap A\cap C$ and $j\in S\cap A\cap D$.  Let $y_B$ be any string such that $\braket{y_B}{\psi} \ne 0$ and $(y_B)_i = 0$ for some $i\in S\cap B$.  Such a string exists by the assumption that $G$ does not simplify on $\ket{\psi}$.  Now letting $y := (y_A)_{|A} \cup (y_B)_{|B}$ as in Case~1, we observe that $y$ is a test string and that
\[ \braket{y}{\psi} = \braket{y_{|A}\cup y_{|B}}{\psi} = \braket{y_{|A}}{\psi}_A \braket{y_{|B}}{\psi}_B \ne 0\;. \]
This concludes the proof of Case~2.

Assume Case~3.  Let $y_A$ be any string such that $\braket{y_A}{\psi} \ne 0$ and $(y_A)_i = 0$ for some $i\in S\cap A$.  Let $y_B$ be any string such that $\braket{y_B}{\psi} \ne 0$ and $(y_B)_i = 0$ for some $i\in S\cap B$.  Both strings exist by the assumption that $G$ does not simplify on $\ket{\psi}$.  Now letting $y := (y_A)_{|A} \cup (y_B)_{|B}$ as in Cases~1 and 2, we observe that $y$ is a test string and that
\[ \braket{y}{\psi} = \braket{y_{|A}\cup y_{|B}}{\psi} = \braket{y_{|A}}{\psi}_A \braket{y_{|B}}{\psi}_B \ne 0\;. \]
This concludes the proof of Case~3.
\end{proof}

\section{Acknowledgments}

We would like to thank Daniel Grier, Daniel Pad\'{e}, and Thomas Thierauf for interesting and useful discussions.

\bibliographystyle{plain}
\bibliography{../../research/bib/master}

\newpage

\appendix

\section{Calculations}
\label{sec:calculations}

\begin{lemma}\label{lem:4-sets}
Let $\eta\in\complexes$ be such that $\eta \ne 1$.  Let complex numbers $a_{jk}$, $b_{jk}$, $c_{jk}$, and $d_{jk}$ for $j,k\in\{0,1\}$ satisfy
\begin{align}\label{eqn:4-vals-all-1}
a_{11}b_{11} &= \eta\, c_{11}d_{11}\;, \\ \label{eqn:4-vals-not-all-1}
a_{jk}b_{\ell m} &= c_{j\ell}d_{km}
\end{align}
for all $j,k,\ell,m\in\{0,1\}$ such that $jk\ell m = 0$.  If $a_{11}$ and $b_{11}$ are nonzero, then either $c_{00} = c_{01} = c_{10} = 0$ or $d_{00} = d_{01} = d_{10} = 0$.  It follows that for all $r,s\in\{0,1\}$,
\begin{equation}\label{eqn:4-vals-rs}
a_{r0}b_{0s} = a_{0r}b_{s0} = c_{r0}d_{0s} = c_{0r}d_{s0} = 0\;.
\end{equation}
\end{lemma}

\begin{proof}
If $a_{11}b_{11} \ne 0$, then by Equation~(\ref{eqn:4-vals-all-1}) we have $\eta$, $c_{11}$, and $d_{11}$ are all nonzero as well.  Letting $j := k := 1$ in Equations~(\ref{eqn:4-vals-all-1},\ref{eqn:4-vals-not-all-1}), we can solve for each $b_{\ell m}$ in terms of the other quantities:
\begin{align*}
b_{00} &= c_{10}d_{10}/a_{11} & b_{01} &= c_{10}d_{11}/a_{11} \\
b_{10} &= c_{11}d_{10}/a_{11} & b_{11} &= \eta\,c_{11}d_{11}/a_{11}
\end{align*}
Substituting these values into the other 12 equations (where $jk = 0$) and simplifying, we get
\begin{align*}
a_{00}c_{10}d_{10} &= a_{11}c_{00}d_{00} & a_{01}c_{10}d_{10} &= a_{11}c_{00}d_{10} & a_{10}c_{10}d_{10} &= a_{11}c_{10}d_{00} \\
a_{00}c_{10}d_{11} &= a_{11}c_{00}d_{01} & a_{01}c_{10} &= a_{11}c_{00} & a_{10}c_{10}d_{11} &= a_{11}c_{10}d_{01} \\
a_{00}c_{11}d_{10} &= a_{11}c_{01}d_{00} & a_{01}c_{11}d_{10} &= a_{11}c_{01}d_{10} & a_{10}d_{10} &= a_{11}d_{00} \\
\eta\,a_{00}c_{11}d_{11} &= a_{11}c_{01}d_{01} & \eta\,a_{01}c_{11} &= a_{11}c_{01} & \eta\,a_{10}d_{11} &= a_{11}d_{01}
\end{align*}
Using the three equations on the bottom row, we solve for $a_{00}$, $a_{01}$, and $a_{10}$:
\begin{align*}
a_{00} &= \frac{a_{11}c_{01}d_{01}}{\eta\,c_{11}d_{11}} & a_{01} &= \frac{a_{11}c_{01}}{\eta\,c_{11}} & a_{10} &= \frac{a_{11}d_{01}}{\eta\,d_{11}}
\end{align*}
and plug these values into the remaining nine equations and simplify to get
\begin{align*}
c_{01}c_{10}d_{01}d_{10} &= \eta\,c_{00}c_{11}d_{00}d_{11} & c_{01}c_{10}d_{01} &= \eta\,c_{00}c_{11}d_{10} & c_{10}d_{01}d_{10} &= \eta\,c_{10}d_{00}d_{11} \\
c_{01}c_{10}d_{01} &= \eta\,c_{00}c_{11}d_{01} & c_{01}c_{10} &= \eta\,c_{00}c_{11} & c_{10}d_{01} &= \eta\,c_{10}d_{01} \\
c_{01}d_{01}d_{10} &= \eta\,c_{01}d_{00}d_{11} & c_{01}d_{10} &= \eta\,c_{01}d_{10} & d_{01}d_{10} &= \eta\,d_{00}d_{11}
\end{align*}
Noting that $\eta \ne 1$, from the last equation on the second row and the second equation on the last row we get
\[ c_{10}d_{01} = c_{01}d_{10} = 0\;. \]
Substituting these values into the equations on the first row and first column, we get for the seven remaining equations
\begin{align*}
c_{00}d_{00} &= 0 & c_{00}d_{10} &= 0 & c_{10}d_{00} &= 0 \\
c_{00}d_{01} &= 0 & c_{01}c_{10} &= \eta\,c_{00}c_{11} & & \\
c_{01}d_{00} &= 0 & & & d_{01}d_{10} &= \eta\,d_{00}d_{11}
\end{align*}
Suppose $c_{00} \ne 0$.  Then the top left equation and its two adjacent equations imply $d_{00} = d_{01} = d_{10} = 0$.  Symmetrically, if $d_{00} \ne 0$, then the corner equations imply $c_{00} = c_{01} = c_{10} = 0$.  Combining this fact with Equation~(\ref{eqn:4-vals-not-all-1}) gives us Equation~(\ref{eqn:4-vals-rs}).
\end{proof}

\begin{remark}
The proof above did not use the two equations $c_{01}c_{10} = \eta\,c_{00}c_{11}$ and $d_{01}d_{10} = \eta\,d_{00}d_{11}$.  They show that $c_{00}$ is uniquely determined by the other $c$'s and $\eta$.  Also, if $c_{00} \ne 0$, then $c_{01} \ne 0$ and $c_{10} \ne 0$, and conversely.  Similarly for the $d$'s.
\end{remark}

\begin{lemma}\label{lem:3-sets}
Let $\eta\in\complexes$ be such that $\eta \ne 1$.  Let complex numbers $a_{jk}$, $b_{j}$, $c_{j}$, and $d_{jk}$ for $j,k\in\{0,1\}$ satisfy
\begin{align}\label{eqn:3-vals-all-1}
a_{11}b_{1} &= \eta\, c_{1}d_{11}\;, \\ \label{eqn:3-vals-not-all-1}
a_{jk}b_{m} &= c_{j}d_{km}
\end{align}
for all $j,k,m\in\{0,1\}$ such that $jkm = 0$.  If $a_{11}$ and $b_{1}$ are nonzero, then either $c_{0} = 0$ or $d_{00} = d_{10} = 0$.  Thus
\begin{equation}\label{eqn:3-vals-rs}
a_{00}b_{0} = c_{0}d_{00} = a_{01}b_{0} = c_{0}d_{10} = 0\;.
\end{equation}
\end{lemma}

\begin{proof}
If $a_{11}b_{1} \ne 0$, then $\eta$, $c_{1}$, and $d_{11}$ are all nonzero as well.  Letting $j := k := 1$ in Equations~(\ref{eqn:3-vals-all-1},\ref{eqn:3-vals-not-all-1}), we can solve for each $b_{m}$ in terms of the other quantities:
\begin{align*}
b_{0} &= c_{1}d_{10}/a_{11} & b_{1} &= \eta\,c_{1}d_{11}/a_{11}
\end{align*}
Substituting these values into the other six equations (where $jk = 0$) and simplifying, we get
\begin{align*}
a_{00}c_{1}d_{10} &= a_{11}c_{0}d_{00} & a_{01}c_{1}d_{10} &= a_{11}c_{0}d_{10} & a_{10}d_{10} &= a_{11}d_{00} \\
\eta\,a_{00}c_{1}d_{11} &= a_{11}c_{0}d_{01} & \eta\,a_{01}c_{1} &= a_{11}c_{0} & \eta\,a_{10}d_{11} &= a_{11}d_{01}
\end{align*}
Using the three equations on the bottom row, we solve for $a_{00}$, $a_{01}$, and $a_{10}$:
\begin{align*}
a_{00} &= \frac{a_{11}c_{0}d_{01}}{\eta\,c_{1}d_{11}} & a_{01} &= \frac{a_{11}c_{0}}{\eta\,c_{1}} & a_{10} &= \frac{a_{11}d_{01}}{\eta\,d_{11}}
\end{align*}
and plug these values into the remaining three equations and simplify to get
\begin{align*}
c_{0}d_{01}d_{10} &= \eta\,c_{0}d_{00}d_{11} & c_{0}d_{10} &= \eta\,c_{0}d_{10} & d_{01}d_{10} &= \eta\,d_{00}d_{11}
\end{align*}
Noting that $\eta \ne 1$, from the middle equation we get that
\begin{equation}\label{eqn:c0d10}
c_{0}d_{10} = 0\;.
\end{equation}
Substituting these values into the first equation gives
\begin{align}\label{eqn:c0d00}
c_0d_{00} &= 0 & d_{01}d_{10} &= \eta\,d_{00}d_{11}
\end{align}
If $c_{0} \ne 0$, then $d_{00} = d_{10} = 0$ by (\ref{eqn:c0d10},\ref{eqn:c0d00}).  Combining this fact with Equation~(\ref{eqn:3-vals-not-all-1}) gives us Equation~(\ref{eqn:3-vals-rs}).
\end{proof}

\begin{remark}
The unused second equation of (\ref{eqn:c0d00}) shows that $d_{00}$ is uniquely determined by the other $d$'s and $\eta$.  Also, if $d_{00} \ne 0$, then $d_{01} \ne 0$ and $d_{10} \ne 0$, and conversely.
\end{remark}

\begin{lemma}\label{lem:2-sets}
Let $\eta\in\complexes$ be such that $\eta \ne 1$.  Let complex numbers $a_{j}$, $b_{j}$, $c_{j}$, and $d_{j}$ for $j\in\{0,1\}$ satisfy
\begin{align}\label{eqn:2-vals-all-1}
a_{1}b_{1} &= \eta\, c_{1}d_{1}\;, \\ \label{eqn:2-vals-not-all-1}
a_{j}b_{m} &= c_{j}d_{m}
\end{align}
for all $j,m\in\{0,1\}$ such that $jm = 0$.  If $a_{1}$ and $b_{1}$ are nonzero, then
\begin{equation}\label{eqn:2-vals-rs}
a_{0}b_{0} = c_{0}d_{0} = 0\;.
\end{equation}
\end{lemma}

\begin{proof}
If $a_{1}b_{1} \ne 0$, then $\eta$, $c_{1}$, and $d_{1}$ are all nonzero as well.  Letting $j := 1$ in Equations~(\ref{eqn:2-vals-all-1},\ref{eqn:2-vals-not-all-1}), we can solve for each $b_{m}$ in terms of the other quantities:
\begin{align*}
b_{0} &= c_{1}d_{0}/a_{1} & b_{1} &= \eta\,c_{1}d_{1}/a_{1}
\end{align*}
Substituting these values into the other two equations (where $j = 0$) and simplifying, we get
\begin{align*}
a_{0}c_{1}d_{0} &= a_{1}c_{0}d_{0} & \eta\,a_{0}c_{1} &= a_{1}c_{0}
\end{align*}
We use the second equation to solve for $a_{0}$:
\begin{align*}
a_{0} &= \frac{a_{1}c_{0}}{\eta\,c_{1}}
\end{align*}
and plug this value into the first equation and simplify to get
\begin{align*}
c_{0}d_{0} &= \eta\,c_{0}d_{0}
\end{align*}
Noting that $\eta \ne 1$, we get that
\begin{equation}\label{eqn:c0d0}
c_{0}d_{0} = 0\;.
\end{equation}
Combining Equations~(\ref{eqn:2-vals-not-all-1},\ref{eqn:c0d0}) gives us Equation~(\ref{eqn:2-vals-rs}).
\end{proof}

\end{document}